\documentclass[letterpaper,11pt]{article}
\usepackage{geometry}
\usepackage{float}
\usepackage[T1]{fontenc}
\usepackage{lmodern}

\usepackage{amssymb}

\usepackage{natbib}

\title{The Importance of Knowing the Arrival Order in Combinatorial Bayesian Settings}
\author{Tomer Ezra\thanks{Sapienza University of Rome, Email:  \texttt{ezra@diag.uniroma1.it}} \and  Tamar Garbuz}
\date{}
\newcommand{\cf}{\mathcal{F}}
\newcommand{\dists}{\mathcal{D}}
\newcommand{\dist}{D}

\usepackage{xcolor}

\usepackage{mathtools}
\usepackage{bbm}

\newcommand{\inst}{\mathcal{I}}
\newcommand{\elements}{E}

\newcommand{\alg}{ALG}
\newcommand{\opt}{OPT}

\newcommand{\pigood}{\pi_1}
\newcommand{\pibad}{\pi_0}
\newcommand{\piind}[2]{\pi_{\mathbbm{1}_{ #2 \in r_#1  }}}

\newcommand{\concat}{\cdot}

\newcommand*{\eqdef}{\stackrel{\text{def}}{=}}

\newcommand{\ificalp}[1]{}

\usepackage{amsthm}

\newtheorem{theorem}{Theorem}[section]

\newtheorem{claim}[theorem]{Claim}
\newcommand{\ordist}{F}
\newcommand{\Ex}{E}
\newtheorem{example}{Example}[section]

\begin{document}

\maketitle
\begin{abstract}
    We study the measure of order-competitive ratio introduced by \ificalp{Ezra et al.} \citet{ezra2023next} for online algorithms in Bayesian combinatorial settings.
In our setting, a decision-maker observes a sequence of elements that are associated with stochastic rewards that are drawn from known priors, but revealed one by one in an online fashion.
The decision-maker needs to decide upon the arrival of each element whether to select it or discard it (according to some feasibility constraint), and receives the associated rewards of the selected elements. 
The order-competitive ratio is defined as the worst-case ratio (over all distribution sequences) between the performance of the best order-unaware and order-aware algorithms, and quantifies the loss incurred due to the lack of knowledge of the arrival order.

\ificalp{Ezra et al.} \citet{ezra2023next} showed how to design algorithms that achieve better approximations with respect to the new benchmark (order-competitive ratio) in the single-choice setting, which raises the natural question of whether the same can be achieved in combinatorial settings. In particular, whether it is possible to achieve a constant approximation with respect to the best online algorithm for downward-closed feasibility constraints, whether $\omega(1/n)$-approximation is achievable for general (non-downward-closed) feasibility constraints, or whether a convergence rate to $1$ of $o(1/\sqrt{k})$ is achievable for the multi-unit setting.
We show, by devising novel constructions that may be of independent interest, that for all three scenarios, the asymptotic lower bounds with respect to the old benchmark, also hold with respect to the new benchmark.

\end{abstract}

\section{Introduction}
We revisit the prophet inequality problem in combinatorial settings.
In the prophet inequality setting \citep{krengel1977semiamarts, krengel1978semiamarts, samuel1984comparison} there is a sequence of boxes, each contains a stochastic reward drawn from a known distribution. The rewards are revealed one by one to a decision-maker, that needs to decide whether to take the current reward, or continue to the next box. The decision-maker needs to make the decisions in an immediate and irrevocable way, where her goal is to maximize her expected selected reward. 
The most common performance measure for the analysis of the decision-maker policy is the competitive-ratio, which is the ratio between the expectation of the selected reward and the expected maximum reward. That is, the decision-maker is evaluated by comparison to a ``prophet'' who can see into the future and select the maximal reward. 
This framework has been extended to combinatorial settings, where the decision-maker is allowed to select a set of boxes (instead of only one) under some predefined feasibility constraints, such as multi-unit \citep{hajiaghayi2007automated,DBLP:journals/siamcomp/Alaei14}, matroids \citep{KleinbergW19}, matching \citep{feldman2014combinatorial,ezra2022prophet}, and downward-closed (or even general) feasibility constraints \citep{DBLP:conf/stoc/Rubinstein16}.
    
A recent line of work studied the (combinatorial) prophet setting when instead of comparing to the best offline optimum (or the ``prophet''), they compare against the best online algorithm \citep{niazadeh2018prophet,papadimitriou2021online,braverman2022max}, and showed how to achieve tighter approximations compared to the best online algorithms.    

Recently, \ificalp{Ezra et al.} \citet{ezra2023next} suggested the benchmark termed ``order-competitive ratio'' defined as the worst-case ratio (over all distribution sequences) between the expectations of the best \textit{order-unaware} algorithm and the best \textit{order-aware} algorithm. Thus, the order-competitive ratio quantifies the loss that is incurred to the algorithm due to an unknown arrival order.
\ificalp{Ezra et al.} \citet{ezra2023next}  showed that for the single-choice prophet inequality setting, it is possible to achieve $1/\phi$-approximation with respect to the new benchmark (where $\phi$ is the golden ratio).
In particular, they showed a separation between what adaptive and static algorithms can achieve with respect to the new benchmark, while with respect to the optimum offline, there is no such separation as a static threshold can achieve the tight approximation of $1/2$.

The question that motivates this paper is whether one can achieve improved approximations for the new benchmark in combinatorial settings.
In particular, whether it is possible to achieve a constant approximation with respect to the best online algorithm for downward-closed feasibility constraints, whether $\omega(1/n)$-approximation is achievable for general (non-downward-closed) feasibility constraints, or whether a convergence rate to $1$  of $o(1/\sqrt{k})$ is achievable for the multi-unit setting.

\subsection{Our Contribution, Techniques, and Challenges}
We study this question in three natural and generic combinatorial structures: $k$-uniform matroid (also known as multi-unit), downward-closed, and arbitrary (not downward-closed) feasibility constraints.

The first scenario we consider is downward-closed feasibility constraints. We first revisit the example in \citep{DBLP:conf/stoc/Rubinstein16} that is based on the upper bound of \citet{babaioff2007matroids} for a different setting, that shows that no algorithm can achieve an approximation of $\omega\left(\frac{\log\log(n)}{\log(n)}\right)$:
\begin{example}[\citep{babaioff2007matroids}]\label{ex:dc}
Consider a set of $n=2^{2^k}$ elements, that are partitioned into $2^{2^k-k}$ parts, each of size $2^k$. The reward of each element is $1$ with probability $2^{-k}$ and $0$ otherwise. The feasibility constraint is such that the decision-maker is allowed to select elements from at most one part of the partition. The elements arrive in an arbitrary order. It is easy to verify that the expected value of the prophet is $\Omega(2^k)$, since it is a maximum of $2^{2^k-k}$ random variables that are distributed according to $Bin(2^k,2^{-k})$.
On the other hand, no online algorithm can have an expected reward of more than $2$, since once the algorithm decides to select an element (with a value at most $1$), then the expectation of the sum of the remaining feasible elements is bounded by $1$.
\end{example}
As can be observed in Example~\ref{ex:dc}, the instance is constructed in a way that no online algorithm (order aware or unaware) can achieve an expected reward of more than $2$, while achieving an expected reward of $1$ is trivial. 
Thus, it fails to show a gap between what order-aware and order-unaware algorithms can achieve.
This leads us to our first result.

\vspace{-0.1in}

\paragraph*{Result A (Theorem~\ref{thm:hardness}):} No order-unaware algorithm can achieve an approximation of $\omega\left(\frac{\log\log(n)}{\log(n)}\right)$ with respect to the best order-aware online algorithm.
\vspace{0.2in}

To show Result A, we need to develop an entirely different construction than the one used in \citep{babaioff2007matroids}. Their construction is such that once the online algorithm selects an arbitrary element, it eliminates all the flexibility that the algorithm had in choosing elements due to the feasibility constraint.
All attempts that are only based on the construction of the feasibility constraint, are destined to fail since the feasibility constraint will influence both the order-aware and order-unaware algorithms in the same way.
Thus, we construct a pair of a feasibility constraint and a distribution over arrival orders. Our elements are partitioned into $k$ layers, and within each layer, the elements are symmetric (with respect to the feasibility constraint). An algorithm needs to select at most one element of each layer. The difference between the elements within the layers, is the role with respect to the arrival order, which draws half of them to be ``good'', and half of them to be ``bad''. ``Good'' elements, are such that the best order-aware algorithm does not lose a lot by choosing them, and ``bad'' elements, are such that the best order-aware algorithm does lose a lot by choosing them.
An order-aware algorithm can distinguish between ``good'' and ``bad'' elements and can always choose the ``good'' ones, while an order-unaware algorithm cannot distinguish between them, therefore cannot do better than guessing and thus it will guess a ``bad'' one after a constant number of layers in expectation.

\vspace{0.1in}

The second scenario that we consider is of arbitrary feasibility constraints. For this problem with respect to the best offline algorithm as a benchmark, \citet{DBLP:conf/stoc/Rubinstein16} showed that no online algorithm can achieve a competitive-ratio of $\omega\left(\frac{1}{n}\right)$. Achieving a competitive-ratio of $\frac{1}{n}$ can be done trivially by selecting the feasible set with the maximal expectation.
We next revisit the example in \cite{DBLP:conf/stoc/Rubinstein16} that shows that no online algorithm can achieve an approximation of 
$\omega\left(\frac{1}{n}\right)$.

\begin{example}[\citep{DBLP:conf/stoc/Rubinstein16}]\label{ex:ndc}
    Consider an instance with $n=2k$ elements, where the collection of feasible sets is $\left\{\{i,i+k\} \mid i \in [k] \right\}$. The elements arrive according to the order $1,\ldots,n$, and the value of each element in $[k]$ is deterministically $0$, while the value of each element in $\{k+1,\ldots,n\}$ is $1$ with probability $\frac{1}{n}$, and $0$ otherwise.
    The prophet receives a value of $1$ if one of the elements of the second type has a non-zero value, which happens with a constant probability.
    Every online algorithm must select exactly one element among the elements of the first type, which restricts the algorithm to select a specific element of the second type, therefore every online algorithm has an expected value of $\frac{1}{n}$.
\end{example}

As can be observed in Example~\ref{ex:ndc}, the instance is constructed with a fixed order, and the optimal algorithm for this feasibility constraint (even for every arrival order), is to discard all zero-value elements and select all elements with a value of $1$ as long as there is a way to complete the chosen set to a feasible set. This algorithm is an order-unaware algorithm, and therefore this construction does not induce a separation between what order-unaware and order-aware algorithms can achieve. This leads us to our second result.

\vspace{-0.1in}

\paragraph*{Result B (Theorem~\ref{thm:non-downward}):} No order-unaware algorithm can achieve an approximation of $\frac{1+\Omega(1)}{n}$ with respect to the best order-aware online algorithm.
\vspace{0.2in}

Our result improves upon the result in \cite{DBLP:conf/stoc/Rubinstein16} in two dimensions. First, our result is with respect to the tighter benchmark of the best online algorithm rather than the best offline algorithm. Second, our upper bound matches the lower bound, up to low-order terms (and not just up to a constant).

To show Result B, we create three types of elements: The first type of elements is of elements with a value of $1$ with a small probability. Almost all elements are of this type, and the utility of the instance comes from these elements. The feasibility constraint requires to select exactly one of these elements. The elements of the other two types have a deterministic value of $0$, and their role is to limit the ability of the algorithm to select elements of the first type.
The feasibility constraint is such that for each subset of elements of type 2, and each element of type 1, there is exactly one subset of elements of type 3 such that their union is feasible. The order of arrival is such that in Phase 1, the elements of type 2 arrive, in Phase 2, most of the elements of type 3, in Phase 3, the elements of type 1 arrive, and in Phase 4, the remaining (few) elements of type 3 arrive.
For exactly one subset $X$ of the elements of type 2,  it holds that: for each element $e$ of type 1, there is a subset $X_e$ of elements of type 3 that arrive in Phase 4, such that $X\cup \{e\} \cup X_e$ is a feasible set. For all other choices of $X$, there are at most a few feasible elements of type 3 that arrive at Phase 4, which restricts the algorithm to choose only among a few elements of type 1.
The only way to ``catch'' the value of all the elements of type 1, is to correctly guess the unique good subset $X$ of type 2 with this special property. An order-aware algorithm can always guess it correctly as this information can be derived from the arrival order (since it knows the partition of elements of type 3 between Phase 2 and Phase 4), while an order-unaware cannot guess the correct subset with high enough probability, and therefore it loses a factor of $\frac{1}{n}$ in the approximation.

\vspace{0.1in}

The third scenario that we consider is of $k$-capacity feasibility constraints. For this problem with respect to the best offline algorithm as a benchmark, \ificalp{Hajiaghayi et al.} \citet{hajiaghayi2007automated} showed that no online algorithm can achieve a competitive-ratio of $1-o\left(\frac{1}{\sqrt{k}}\right)$. Achieving a competitive-ratio of $1-\Theta\left(\frac{1}{\sqrt{k}}\right)$ with respect to the best-offline is achieved by \ificalp{Alaei} \citet{DBLP:journals/siamcomp/Alaei14}.

Our last result shows, that one cannot achieve an order-competitive ratio that converges to $1$ in a faster rate (up to a constant).

\vspace{-0.1in}

\paragraph*{Result C (Theorem~\ref{thm:kunit}):} No order-unaware algorithm can achieve an approximation of $1-o\left(\frac{1}{\sqrt{k}}\right)$ with respect to the best order-aware online algorithm.

%\vspace{0.2in}

To show Theorem~C, we construct an instance with three types of elements. The first type is with a deterministic low value, the second type is with a deterministic mid-value, and the third type is randomized, with a probability half of being high, and a probability half of being zero. The order of arrival is such that all the type~2 elements arrive first, and then either all elements with type~1 arrive before all elements of type~3 which is considered the ``bad'' order, or vice versa which is the ``good'' order. An algorithm that knows whether it is a good order or a bad order, can adapt the number of elements of type~2 to choose in an optimal way, while an algorithm that does not know the order needs to commit to selecting elements of type~2 before any information regarding the order is revealed. Our analysis then follows by balancing the low, mid, and high values in a way that an order-unaware  algorithm that commits to selecting a certain amount of elements of type 2, will be far from the optimal order-aware algorithm for one of the two arrival orders.

\subsection{Further Related Work}
\paragraph*{Comparing to the best online.}
Our work is largely related to a line of research that examines alternative benchmarks for the best offline benchmark, and in particular, comparing its performance to the best online algorithm
%Mostly related to our work, is the line of research considering alternative benchmarks to the best offline benchmark, and in particular, comparing the performance to the best online algorithm
\citep{niazadeh2018prophet,kessel2022stationary,papadimitriou2021online,saberi2021greedy,braverman2022max,ezra2023next}.
One example, \ificalp{Niazadeh et al.}\cite{niazadeh2018prophet} showed that the original tight prophet inequality bounds comparing the single-pricing with the optimum offline are tight even when comparing to the optimum online as a benchmark (both for the identical and non-identical distributions).
Another example is that \ificalp{Papademitriou et al.} \cite{papadimitriou2021online} studied the online stochastic maximum-weight matching problem under vertex arrivals, and presented a polynomial-time algorithm which approximates the optimal online algorithm within a factor of 0.51, which was later improved by \ificalp{Saberi and Wajc} \cite{saberi2021greedy} to 0.526, and to $1-1/e$ by \ificalp{Braverman et al.} \cite{braverman2022max}.
\ificalp{Kessel et al.} \cite{kessel2022stationary} studied a continuous and infinite time horizon counterpart to the classic prophet inequality, term the stationary prophet inequality problem. They showed how to design pricing-based policies which achieve a tight $\frac{1}{2}$-approximation to the optimal offline policy, and a better than $(1-1/e)$-approximation of the optimal online policy. 

%\paragraph*{Alternative benchmarks in prophet settings.}
%\ificalp{Ezra et al. 2022 } 
%\citet{ezra2022EOR} studied the benchmark of the expectation of the ratio between the value of the algorithm and the value of the prophet, and established a two-way black-box reduction from the new benchmark to the old one that loses at most a constant.

\paragraph*{Prophet in combinatorial settings.}
Another line of work, initiated by \ificalp{Kennedy} \citet{kennedy1985optimal,kennedy1987prophet}, and \ificalp{Kertz}\citet{kertz1986comparison},
extends the single-choice optimal stopping problem to multiple-choice settings. Later work extended it to additional combinatorial settings, including multi-unit \citep{hajiaghayi2007automated,DBLP:journals/siamcomp/Alaei14} matroids~\citep{KleinbergW19,azar2014prophet}, 
polymatroids~\citep{dutting2015polymatroid}, matching~\citep{GravinW19,ezra2022prophet}, combinatorial auctions \citep{feldman2014combinatorial,dutting2020prophet,dutting2020log}, and downward-closed (and beyond) feasibility constrains~\citep{DBLP:conf/stoc/Rubinstein16}. 

\paragraph*{Different arrival models.}
A related line of work studied different assumptions on the arrival order besides the adversarial order~\cite{krengel1977semiamarts,krengel1978semiamarts,samuel1984comparison}.
Examples for such assumptions are random arrival order (also known as the prophet secretary) \cite{esfandiari2017prophet,azar2018prophet,ehsani2018prophet,CorreaSZ21}, and free-order settings, where the algorithm may dictate the arrival order~\cite{beyhaghi2018improved,AgrawalSZ20,PengT22}. Another recent study related to the arrival order has shown that for any arrival order $\pi$, the better of $\pi$ and the reverse order of $\pi$ achieves a competitive-ratio of at least the inverse of the golden ratio  \citep{arsenis2021constrained}.

\section{Model}
An instance $\inst$ of our setting is defined by a triplet $\inst=(\elements,\dists,\cf)$ where $\elements$ is the ground set of elements, each element $e\in \elements$ is associated with a distribution $\dist_e\in \dists$, and a feasibility constraint $\cf \subseteq 2^\elements$ over the set of elements (where $\cf\neq \emptyset$).
The elements arrive one by one. Upon the arrival of element $e$, its identity is revealed, and a value $v_e$ is drawn independently from the underlying  distribution $\dist_e$.
 We call an instance $\inst$  binary if for every element $e\in\elements$, the support of $\dist_e$ is $\{0,1\}$.

A decision-maker, who observes the sequence of elements and their values, needs to decide upon the arrival of each element whether to select it or not subject to the feasibility constraint $\cf$, which asserts that the set that is chosen at the end of the process (after all elements arrive) must belong to $\cf$. Another interpretation of the feasibility constraint, is that the decision-maker must select (respectively discard) element $e$ if all feasible sets that agree with all previous decisions before the arrival of element $e$, contain (respectively do not contain) element $e$. A feasibility constraint $\cf$ is called downward-closed if for every set $S\in \cf$, and a subset $T\subseteq S$, then $T$ must be in $\cf$. For downward-closed feasibility constraints, discarding elements is always feasible. The decision-maker's utility is the sum of the values of the selected elements.

We say that a decision-maker (or algorithm) is \textit{order-unaware} if she does not know the arrival order of the elements in advance, and needs to make decisions with uncertainty regarding the order of the future arriving elements. 
We say that a decision-maker (or algorithm) is \textit{order-aware}, if she knows the order of arrival of the elements in advance, and can base her decisions on this information. 
Given an instance $\inst$, an order of arrival of the elements $\pi$, and an algorithm $\alg_\inst$ (that might be order-unaware, or order-aware), we will denote the expected utility of $\alg_\inst$ for arrival order $\pi$, by $\alg_\inst(\pi)$.
Given an instance $\inst$ and an arrival order $\pi$, we will denote the order-aware algorithm with the maximal expected utility by $\opt_{\inst,\pi}$, i.e., $\opt_{\inst,\pi}\eqdef \arg\max_{\alg_\inst} \alg_{\inst}(\pi)$.

We want to quantify the importance of knowing the order in advance,  and to do so, we use the measure of \textit{order-competitive ratio} proposed by \citet{ezra2023next} for the case of choosing a single element (i.e, $\cf = \{S\subseteq E \mid |S|\leq 1\}$).
Given an instance $\inst$, the order-competitive ratio of an order-unaware algorithm $\alg_\inst$, denoted by $\rho(\inst, \alg_\inst)$ is  
\begin{equation}
    \rho(\inst, \alg_\inst) \eqdef \min_{\pi } \frac{\alg_\inst(\pi)}{\opt_{\inst,\pi}(\pi)}. \label{eq:ocr}
\end{equation}

We use $[j]$ to denote the set  $\{1,\ldots,j\}$. Given two partial orders $\pi^1 = (e^1_1,\ldots,e^1_{k_1})$, and $ \pi^2=(e^2_1,\ldots,e^2_{k_2})$ over two disjoint subsets of elements $\elements_1,\elements_2 \subseteq \elements$, we  define the order $\pi^1 \concat \pi^2 \eqdef
(e^1_1,\ldots,e^1_{k_1},e^2_1,\ldots,e^2_{k_2})$.

In this paper, we use the following forms of  Chernoff bound:
\begin{theorem}[Chernoff bound]
For a series of $n$ independent Bernoulli random variables $X_1,\ldots,X_n$, and for $X=\sum_{i=1}^n X_i$ it holds:
\begin{itemize}
    \item For all $ 0 \leq \delta \leq 1$,  $\Pr\left[|X-\Ex[X]| \geq \delta \Ex[X]\right] \leq 2 e^{-\delta^2 \cdot \Ex[X]/3}. $
    \item 
For all $\delta  \geq 0$, $ \Pr\left[X\geq (1+\delta)\Ex[X]\right] \leq e^{-\delta^2 \cdot \Ex[X]/(2+\delta)}. $
\end{itemize}
\end{theorem}

Lastly, for an instance $\inst=(\elements,\dists,\cf)$, and an algorithm $\alg_\inst$ we denote by $\xi(\inst,\alg_\inst)$ the traditional competitive ratio which is 
\begin{equation}
    \xi(\inst, \alg_\inst) \eqdef \min_{\pi } \frac{\alg_\inst(\pi)}{\Ex[\max_{S \in \cf} \sum_{e\in S} v_e] }. \label{eq:cr}
\end{equation}
It is easy to observe, that for every instance $\inst$, and an algorithm $\alg_\inst$, $$ \xi(\inst, \alg_\inst) \leq \rho(\inst, \alg_\inst),$$
thus, every lower bound on the competitive-ratio also applies to the order-competitive ratio (but not vice versa), and any upper bound on the order-competitive ratio also applies to the order-competitive ratio (but not vice versa).

\section{Downward-Closed Feasibility Constraints}

In this section, we show an upper bound on the order-competitive ratio for the family of downward-closed feasibility constraints.
This upper bound also holds with respect to binary instances and matches the best-known upper bound on the competitive-ratio. The current best-known lower bound for the competitive-ratio for downward-closed feasibility constraints of $O\left(\frac{1}{\log^2(n)}\right)$ was proved by \cite{DBLP:conf/stoc/Rubinstein16}, and closing this gap is an open question. 
\begin{theorem}
There exists a constant $\xi>0$ such that for every $n> 2$ there is a (binary) instance $\inst=(\elements,\dists,\cf)$ with $n=|\elements|$ and a downward-closed feasibility constraint $\cf$ in which for every order-unaware algorithm (deterministic or randomized) $\alg_\inst$, it holds that $$     \rho(\inst, \alg_\inst) \leq \frac{\xi \cdot \log \log n}{\log n}.$$  \label{thm:hardness}
\end{theorem}

\begin{proof} %[Proof of Theorem \ref{thm:hardness}]
We assume that $n=\sum_{i=1}^k k^i$ for some even $k$. (Otherwise, we can reduce to the largest $n'\leq n$ that is of this form, by having  $n-n'$ redundant elements.) Notice that \begin{equation}
    k \in \Theta\left(\frac{\log n}{\log\log n}\right),
    \label{eq:defk}
\end{equation} since for $k=\frac{\log n}{2\log\log n}$, it holds that $\sum_{i=1}^k k^i \leq k^{k+1} \leq n$, while for $k=\frac{2\log n}{\log\log n}$, it holds that $\sum_{i=1}^k k^i \geq k^{k} \geq n$ for large enough $n$. 
For every string $s$ of length between $1$ and $k$ where each character is in $[k]$, we define an element $e_s$. We denote by $s_j$ for $j \in [|s|]$ the $j$-th character of the string $s$, moreover, we denote by $s_{[j]}$ the prefix of $s$ of the first $j$ characters.
The set of elements $\elements$ is defined to be $\{e_s \mid s \in  \bigcup_{i=1}^k [k]^{i}\}$. Given a string $s$ and a character $j$ (respectively, another string $s'$), we denote by $sj$ (respectively, $ss'$) the string-concatenation of $j$ (respectively, $s'$) at the end of string $s$.
We say that an element $e_{sj}$ for a string $s$ and $j\in [k]$ is a child of element $e_s$, and that $e_s$ is the parent of $e_{sj}$. (Note that an element can have only one parent, but may have multiple children.) 
The value of all elements are drawn i.i.d. from the distribution  $\dist$ in which $v=1$ with probability $\frac{1}{k}$, and $v=0$ otherwise. Let $\dists \eqdef\{\dist \}_{e \in \elements}$. The feasibility constraint $\cf\eqdef \{S\subseteq \elements \mid \mbox{ for every } e_{s_1},e_{s_2}\in S, \mbox{ if } |s_1| \leq |s_2|, \mbox{ then }   s_1 \mbox{ is a prefix of } s_2 \}$ (in other words, only subsets of a single path from the root to one of the leafs are feasible).
The instance is then $\inst=(\elements,\dists,\cf)$.
It is sufficient to show that for some constant $c>0$, there is a distribution $\ordist$ over the arrival orders, in which the expected utility of every order-unaware algorithm $\alg_\inst$ is at most $c/k$ of the expected utility of the optimal order-aware algorithm.
%when the order of arrival is distributed according to $\ordist$. 
I.e., 
\begin{equation}
\exists c \quad  \forall \alg_\inst \quad\quad\quad
\Ex_{\pi\sim \ordist}[\alg_\inst(\pi)] \leq  \frac{c}{k} \cdot  \Ex_{\pi\sim\ordist}[\opt_{\inst,\pi}(\pi)]. 
\label{eq:expected-opt}
\end{equation}
Equation~\eqref{eq:expected-opt} is sufficient since it shows that for every algorithm $\alg_\inst$ there exists an order $\pi^*$ (in the support of $\ordist$) in which $\alg_\inst(\pi^*) \leq  c/k \cdot  \opt_{\inst,\pi^*}(\pi^*)$, which together with Equation~\eqref{eq:defk} concludes the proof.

We now define the distribution $\ordist$ over the arrival orders. 
We first draw independently for every string $s$ of size between $0$ and $k-3$, a random subset of $[k]$ of size $k/2$, which we will denote by $r_s$. 
Then, the elements arrive in an arrival order defined by the following recursive formulas. We first define for every string $s$ of size between $0$ and $k-1$ and a parameter $i \in  [k-|s|]$:
$$ \pibad(s,i) \eqdef  (ss')_{s'\in [k]^i},$$
and 
$$ \pibad(s) \eqdef \pibad(s,k-|s|) \concat \ldots \concat\pibad(s,1).$$
%For $s$ such that $|s|=k$ we denote by $\pibad(s)=()$ (the empty vector).
We also denote given the random realizations 
$\{r_s\}_{s}$, for every string $s$ of size between $0$ to $k-3$ the arrival order 
$$ \pigood(s) \eqdef (s1,\ldots,sk) \concat \piind{s}{1}(s1)\concat\ldots\concat \piind{s}{j}(sj)\concat \ldots \concat\piind{s}{k}(sk),$$
and for $s$ such that $|s|=k-2$, $$\pigood(s) \eqdef (s1,\ldots,sk)\concat \pibad(s1)\concat\ldots\concat\pibad(sk). $$
The arrival order is then $\pigood(\epsilon)$.

For every element $e_s$, we say that $e_s$ is \textit{good}, if for every $j\in [\min(k-2,|s|)]$, it holds that $s_j \in r_{s_{[j-1]}}$, and \textit{bad} otherwise.
% We will denote the set of good elements by $\goods$.
The order of arrival is illustrated in Figure~\ref{fig:downward}

We first bound from below the RHS of Equation~\eqref{eq:expected-opt}.
\begin{claim} For $c'=\frac{\sqrt{e}}{\sqrt{e}-1} $, it holds that
$   \Ex_{\pi\sim\ordist}[\opt_{\inst,\pi}(\pi)] \geq  k/c' $. \label{cl:aware}
\end{claim}
\begin{proof}
We prove this claim by showing that for every order $\pi$ in the support of $\ordist$, it holds that $\opt_{\inst,\pi}(\pi) \geq \frac{k}{c'}$.
%We first denote $q_s = r_s$ if for every $j\leq |s|$ it holds that $s_j \in r_{s_{[j-1]}}$, otherwise, $q_s=\emptyset$.
Consider an order-aware algorithm (not necessarily $\opt_{\inst,\pi}$) that selects an element $e_s$ if (1) $e_{s}$ is feasible, (2)  $e_s$ is good, and (3) $v_{e_{s}}=1$ or $e_s$ is the last good element to arrive in the set $\{s_{[|s|-1]}j \mid j \in [k]\}$. %, and (4) $|s|=1$ or the father of $e_s$ has been selected.

By the description of the algorithm we know we will only select good elements, and we will select exactly one element from each layer (elements of strings with the same length). The algorithm receives a utility of $1$ from layer $j\in [k]$ if one of the good elements that are the children of the element chosen from layer $j-1$, has a value of $1$. (For elements of layer $1$, it is sufficient that one of the good elements, has a value of $1$.)
%Since there are at least $k/2$ such elements in each layer,
Thus, the expected utility of the algorithm is at least the number of layers, times the probability that one of the (at least) $k/2$ elements has a value of $1$. Therefore $$ \opt_{\inst,\pi}(\pi) \geq k\cdot (1-(1-\frac{1}{k})^{k/2}) \geq k\cdot (1-\frac{1}{\sqrt{e}}) = \frac{k}{c'},$$
which concludes the proof of the claim.
\end{proof}
\begin{figure}[H]
     \centering
     \includegraphics[width=0.77\textwidth]{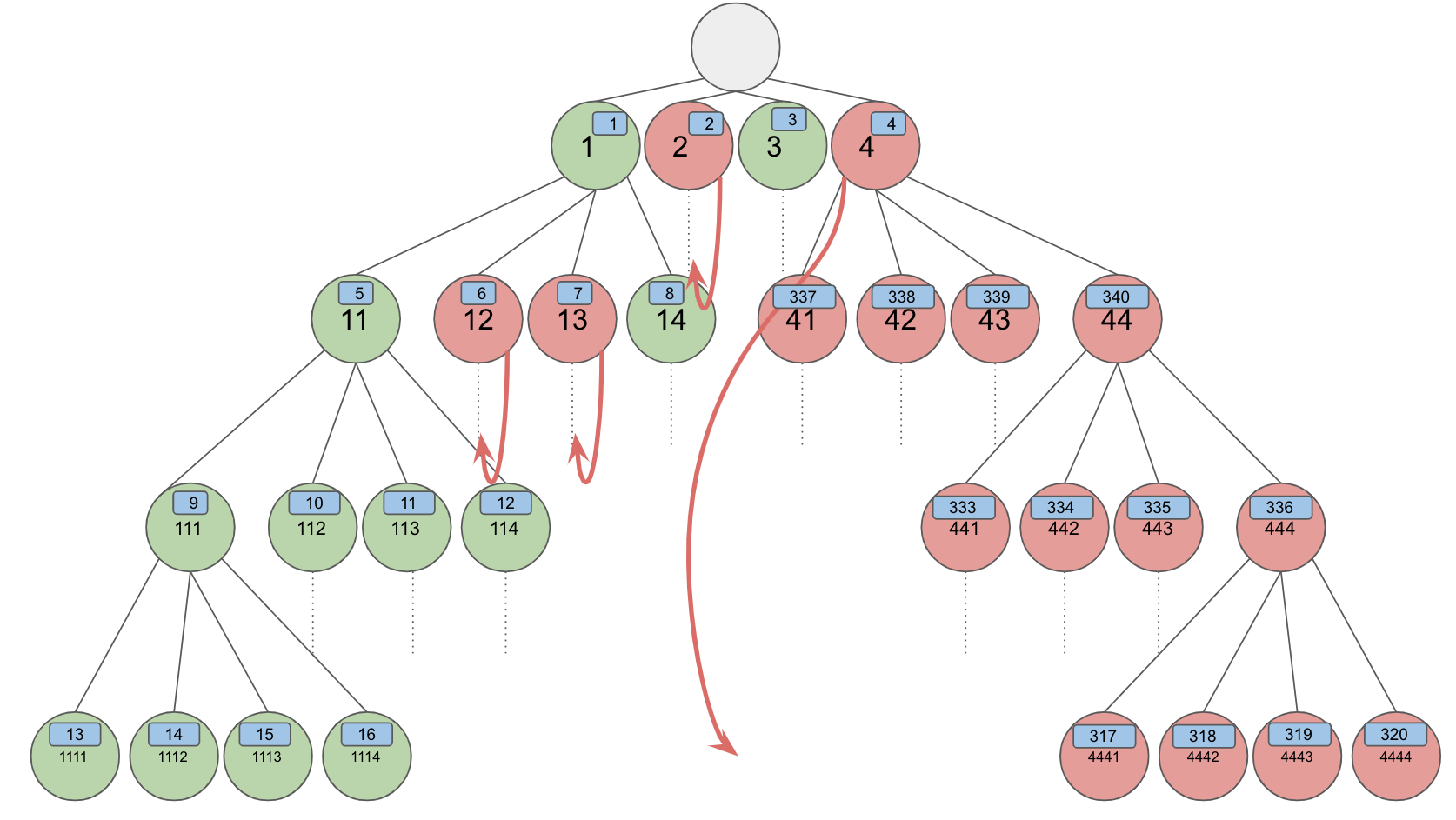}
    \caption{An example of an instance with $k=4$. In this example, there are $n=4+4^2+4^3+4^4 =340$ elements. The elements are partitioned into layers according to the structure of the feasibility constraint. A feasible set under this constraint is a subset of a path from the root to some leaf in the tree (excluding the root which is not an element). The numbers in the center of the circles represent the corresponding string-identity of the elements. 
 Green circles represent good elements, and red circles represent bad elements. For each layer up to the last two layers, all the children of bad elements are bad, and half of the children of good elements are good (and half of them are bad).
 For the last two layers, all the children of good elements are good, and all the children of bad elements are bad.
 In this example, the realizations of the random variables $\{r_s\}_s$ are such that $r_\epsilon = \{1,3\}$, and $r_1=\{1,4\}$. 
The numbers in the blue rectangles represent the arrival time of the element according to the arrival order.  The outgoing red arrows from bad elements that aren't children of bad elements, represent that after the arrival of this element, the next descendant among the sub-tree rooted at this element that is arriving according to the arrival order, is not a child of the element (as happens with good elements) but rather is a leaf, and the order of arrival of this sub-tree is bottom up.}
    \label{fig:downward}
\end{figure}

We next bound from above the LHS of Equation~\eqref{eq:expected-opt}.
\begin{claim}
For every (deterministic or randomized) order-unaware algorithm $\alg_\inst$, it holds that $$ \Ex_{\pi\sim \ordist}[\alg_\inst(\pi)] \leq 5.$$ \label{cl:oblivious}
\end{claim}
\begin{proof}
We analyze the performance of $\alg_\inst$ by partitioning into three types of contributions: (1) good elements, (2) bad elements that are either children of good elements or in the first layer, and (3) bad elements that are children of bad elements.

We first claim that the expected number of elements of type (1) that $\alg_\inst$ selects is at most $2$. 
To show this we can first observe that once a bad element is chosen, then good elements cannot be chosen anymore. After a bad element is chosen, the only elements that can be chosen are the offspring of this element (which are also bad by definition) and the ancestors of the element that haven't arrive yet (which all must be bad).
We next observe, that the algorithm can only select good elements in a strictly increasing order (in the length of their corresponding strings). 
Moreover, for every  element $e_s$ from layer $j$ for  $j\in [k-2]$, that is a child of a good element or is of layer 1, given the information that the algorithm has up to the arrival of element $e_s$, the probability of being good is exactly $1/2$.
%I.e., for all $e\in \elements$, $ \Pr[e \in \goods \mid \father(e)\in \goods\cup\{\epsilon\}  ] = 1/2$.
This is since being good, by definition requires that (1) $e_s$ is not a child of a bad element (which the algorithm knows upon the arrival of $e_s$), and (2) $s_{|s|} \in r_{|s|-1}$, which happens with probability $1/2$.
Thus, each time the algorithm tries to select a good element from the first $k-2$ layers, it can no longer  select additional good elements with probability $1/2$. If the algorithm reaches layer $k-1$ without selecting a bad element, the algorithm can select at most two more good elements.
Therefore if the algorithm tries to ``gamble" and select $\ell$ good elements from the first $k-2$ layers, it selects in expectation at most $\frac{\ell+2}{2^\ell} + \sum_{i=1}^{\ell}  \frac{i-1}{2^i} \leq 2$ good elements from all $k$ layers\footnote{This argument also holds for randomized $\ell$.}.

Second, $\alg_\inst$ can choose at most one element of type (2). This is since in every feasible set, there is at most one such element.  (For every feasible set, only the element that corresponds to the shortest string among the bad ones can be of this type.)

Last, the expected utility of $\alg_\inst$ from elements of type (3) is at most $2$.
This is true since we can observe that once a bad element $e$ that is a child of a bad element is selected, the algorithm can only select elements that are ancestors of $e$.
Since there are less than $k$ such elements, and each can contribute a utility of at most $1/k$ in expectation, the expected utility of elements of this type is less than $2$. (Element $e$ contributes $1$,  and its ancestors contributes less than $k\cdot \frac{1}{k}$.) This concludes the proof.
\end{proof}

The theorem follows by combining Claims~\ref{cl:aware} and \ref{cl:oblivious}, with Equation~\eqref{eq:defk}.
\end{proof}

\section{Non-Downward Closed Feasibility Constraints}
In this section, we present an upper-bound on the order-competitive ratio of arbitrary (non-downward closed) feasibility constraints. This upper-bound holds even with respect to binary instances. This result is tight since achieving an order-competitive ratio of $\frac{1}{n}$ can be done trivially, by an algorithm that selects the set of elements with the maximum expected sum of values among all feasible sets.  Our result also improves the best-known upper bound of the competitive-ratio shown in \citep{DBLP:conf/stoc/Rubinstein16} of $O\left(\frac{1}{n}\right)$ to $\frac{1}{n}+o\left(\frac{1}{n}\right)$.

\begin{theorem}
For every constant $\xi> 1$,  there exists  $n_0 $ such that for every $n\geq n_0$, there exists an instance $\inst=(\elements,\dists,\cf)$ with $ n$ elements (i.e., $n=|\elements|$), in which for every order-unaware algorithm (deterministic or randomized) $\alg_\inst$, it holds that $     \rho(\inst, \alg_\inst) \leq  \frac{\xi}{n}$.  \label{thm:non-downward}
\end{theorem}
\begin{proof} 
We prove that theorem by presenting a construction with $n$ elements, for which no order-unaware algorithm can have an order-competitive ratio of more than $\frac{1}{n}+o(\frac{1}{n})$.
We assume for simplicity that $n= 2^{2x} $ for some integer $x$.
Consider an instance $\inst=(\elements,\dists,\cf)$ in which $\elements = A \cup B \cup C$, where $A=\{a_1,\ldots,a_{k_1}\},$ $B=\{b_1,\ldots,b_{k_2}\}$, and $C=\{c_1\ldots,c_{k_3}\}$ where $k_1=4x$, $k_2=n-\sqrt{n}-4x$, and $k_3=\sqrt{n}$ which sum up to $n$.
The values of all elements in $A \cup C $ are deterministically $0$.
The values of all elements in $B  $ are  $1$ with probability $\frac{1}{n^2}$ and $0$ otherwise.
Let $U_1,\ldots,U_{2^{k_1}}$ be subsets of $C$ which satisfy the conditions from the following claim:
\begin{claim}\label{cl:ui}
There exists $n'_0$ such that  for every $n\geq n'_0$, there exist sets $U_1,\ldots,U_{2^{k_1}}$ such that:
\begin{itemize}
    \item  For all $i\in [2^{k_1}]$, $\log(n') \leq |U_i|\leq 21\cdot \log(n')$.
    \item  For each $j\in [k_3]$, it holds that $|\{i \mid c_j \in U_i\}| \leq  \frac{22\cdot 2^{k_1} \cdot \log(n)}{k_3}$.
    \item  For all $i_1,i_2 \in [2^{k_1}]$ such that $i_1 \neq i_2$, it holds that $|U_{i_1} \cap U_{i_2} | \leq  10$.
\end{itemize}
\end{claim}
\begin{proof}
We prove existence by the probabilistic method. For simplicity of presentation, let $\alpha=10$.
Consider a series of  random variables $X_{ij}$ that indicate whether $c_j\in U_i$, which are drawn independently according to $Ber\left(\frac{(\alpha+1)\cdot \log(n)}{k_3}\right)$. 
Note that for the parameter $\alpha$  and for $n\geq 2^{16}$ this probability is guaranteed to be in $[0,1]$. Let $E^1_i$ be the event that $|U_i| < \log(n)$ or $|U_i| > (2\alpha+1)\cdot \log(n)$ (which is equivalent to $|\sum_{j} X_{ij}- (\alpha+1)\cdot\log(n)|>\alpha\cdot \log(n)$), let $E^2_j$ be the event that $ |\{i \mid c_j \in U_i\}| >  \frac{2(\alpha+1)\cdot 2^{k_1} \cdot \log(n)}{k_3}$ (which is equivalent to $\sum_{i} X_{ij} > \frac{2\cdot(\alpha+1)\cdot 2^{k_1} \cdot \log(n)}{k_3}$), let $E^3_{i_1,i_2}$ be the event that $|U_{i_1}\cap U_{i_2}| > \alpha$ (which is equivalent to $\sum_j X_{i_1 j} \cdot X_{i_2 j } >\alpha$), and let $E$ be the event that one of the formerly defined events occurs, i.e., $E=\left(\bigvee_{i} E^1_i \vee \bigvee_j E^2_j \vee \bigvee_{i_1\neq i_2} E^3_{i_1,i_2}\right)$.
For every $i\in[2^{k_1}]$, it holds that $$\Pr[ E^1_i]  = \Pr\left[|Bin\left(k_3,\frac{(\alpha+1)\cdot\log(n)}{k_3}\right) - (\alpha+1)\cdot\log(n)|> \alpha \cdot \log(n)\right] \leq \frac{2}{n^3},$$
where the inequality is by Chernoff bound. For every $j\in [k_3]$ it holds that  $$\Pr[ E^2_j]  = \Pr\left[Bin\left(2^{k_1},\frac{(\alpha+1)\cdot\log(n)}{k_3}\right) > \frac{2 (\alpha+1)\cdot 2^{k_1}\cdot \log(n)}{k_3}\right]\leq \frac{1}{n^3},$$
where the inequality is by Chernoff Bound. For all $i_1,i_2 \in [2^{k_1}]$, such that $i_1\neq i_2$ it holds that
\begin{eqnarray*}
\Pr[ E^3_{i_1,i_2}]  &= & \Pr\left[Bin\left(k_3,\frac{(\alpha+1)^2\cdot\log^2(n)}{k_3^2}\right) > \alpha \right] \\ 
& \leq & {k_3\choose \alpha+1}  \cdot \left(\frac{(\alpha+1)^2\cdot\log^2(n)}{k_3^2}\right)^{\alpha+1}  \leq   \frac{1}{n^5},
\end{eqnarray*}
where the first inequality holds by the union bound, and the second inequality holds for large enough $n$ (for $n>2^{1000}$).
Thus, by the union bound, the probability that one of the events occurs is $$\Pr[E] \leq 2^{k_1} \cdot \frac{2}{n^3} + k_3 \cdot \frac{1}{n^3} + 2^{2k_1} \cdot \frac{1}{n^5} < 1.$$
Thus, there exist realizations of all $X_{ij}$ in which event $E$ does not occur, which implies the claim.
\end{proof}

Next, we name the subsets of $A$ as $V_1,\ldots,V_{2^{k_1}}$, and  we define $2^{k_1}$ corresponding functions. For each $i\in [2^{k_1}]$, we define an arbitrary injective function  $f_i : [k_2]\rightarrow 2^{U_i}$ (such a function exists since $|U_i| \geq \log(n)\geq \log(k_2)$).
We now define the feasibility constraint $$\cf\eqdef \left\{S \mid \exists i,j \mbox{ such that } S \cap A =V_i~ \wedge ~S\cap B=\{b_j\}   ~\wedge~ S \cap C=f_i(j) \right\}.$$

For every $i$, let $\pi_{i}$ be the arrival order in which the elements arrive in four phases (within each phase, the order can be arbitrary but during the first phase the order should be consistent for all $\pi_i$). Phase~1 is composed of all elements of $A$. Phase~2 is composed of all elements of $C \setminus U_{i}$.
Phase~3 is composed of all elements in $B$, and Phase~4 is composed of all elements in $U_{i}$, i.e.,
$$\pi_i= \left(\underbrace{A}_{\mbox{Phase 1}}, \underbrace{C \setminus U_i}_{\mbox{Phase 2}},\underbrace{B}_{\mbox{Phase 3}},\underbrace{ U_i}_{\mbox{Phase 4}}\right).$$

We next bound from below for every $\pi_i$ the performance of $\opt_{\inst,\pi_i}$ on $\pi_i$.
\begin{claim}\label{cl:opt-non}
For every $\pi_i$, it holds that $\opt_{\inst,\pi_i}(\pi_i) \geq \frac{1}{n}-o\left(\frac{1}{n}\right)$.
\end{claim}
\begin{proof}
Consider the order-aware algorithm, that selects in Phase~1 the subset $V_i$ of $A$. Then in Phase~2 it selects nothing. In Phase~3 it selects the first element $b_j$ of $B$ that its value is $1$ (or the last element of Phase~3, if all of them have values of $0$). In   Phase~4, the algorithm selects the subset $f_i(j)$ of $C$. This is always a feasible set.
The value of this set is $1$, if one of the elements in $B$ has a non-zero value.
The claim then holds since this happens with probability $1-(1-\frac{1}{n^2})^{k_2}  =\frac{1}{n} - o\left(\frac{1}{n}\right)$.
\end{proof}
In order to bound the performance of a randomized algorithm $\alg_\inst$, it is sufficient by Yao's principle to define a distribution $D_\pi$ over arrival orders, and bound the performance of the best deterministic algorithm on the randomized distribution.
Consider the distribution $D_\pi$, where the order $\pi\sim D_\pi$ is $\pi_i $ with probability $\frac{1}{2^{k_1}}$ for every $i\in [2^{k_1}]$.
We next bound from above the performance of any deterministic algorithm $\alg_\inst$.
\begin{claim}\label{cl:alg-non}
For every deterministic algorithm $\alg_\inst$, it holds that $ E_{\pi\sim D _\pi}[ \alg_\inst(\pi)] \leq \frac{1}{n^2} +o\left(\frac{1}{n^2}\right)$.
\end{claim}
\begin{proof}
Let $\alg_\inst$ be an arbitrary deterministic algorithm, then since in Phase~1, the order is constant, $\alg_\inst$ selects deterministically a set $V_{i'} \subseteq A$. We next analyze the performance of $\alg_\inst$ depending on the realized arrival order $\pi_i$. Let $G_{i'} = \{ \pi_i \mid U_i \cap U_{i'}  \neq \emptyset \wedge \pi_i \neq \pi_{i'}\}$.
For every $\pi_i \in G_{i'} $, by Claim~\ref{cl:ui} it holds that $|U_i\cap U_{i'}| \leq 10$, then by the end of Phase~2, $\alg_\inst$ selected a subset of $U_{i'} \setminus U_i$. Since there are at most $10$ elements in $U_{i'} \cap U_i$ that didn't arrive  by the end of Phase~2, there are at most 
$2^{10}$ elements in $B$ that $\alg_\inst$ can select that  lead to a subset of a feasible set. Thus, it holds that $\alg_\inst (\pi_i) 
%\leq 1-(1-\frac{1}{n^2})^{2^{|U_{i} \cap U_{i'}|}} 
\leq \frac{2^{10}}{n^2}$.
For the order of arrival $\pi_{i'}$, it holds that $\alg_\inst (\pi_{i'}) \leq 1-(1-\frac{1}{n^2})^{k_2} \leq \frac{1}{n}$.
Otherwise (for every $\pi_i \neq \pi_{i'}$ such that $\pi_i \notin G_{i'}$), it holds that $U_{i} \cap U_{i'} =\emptyset$, and therefore by the end of Phase~2, there is only one element that $\alg_\inst$ can select which leads to a subset of a feasible set. Thus, $\alg_\inst (\pi_{i}) = \frac{1}{n^2}$.
The set $G_{i'} $ is at most of size $ \sum_{c_j \in U_{i'}} |\{ i \mid c_j \in U_i\} | \leq  21 \cdot\log(n) \cdot \frac{22 \cdot 2^{k_1} \cdot \log(n)}{k_3} = o\left(n^2\right)$, where the inequality is by Claim~\ref{cl:ui}. %since  every element in $U_{i'}$ is in  at most  $\frac{2^{k_1} \cdot \log(n)}{k_3}$ sets  $U_{i}$.
Thus, it holds that $E_{\pi\sim D_\pi}[\alg_{\inst}(\pi)] \leq \frac{1}{2^{k_1}} \cdot \frac{1}{n} + \frac{|G_{i'}|}{2^{k_1}} \cdot \frac{2^{10}}{n^2}  + \frac{2^{k_1}-1-|G_{i'}|}{2^{k_1}} \cdot \frac{1}{n^2}  = \frac{1}{n^2} + o\left(\frac{1}{n^2}\right)$.
\end{proof}
Thus, by combining Claims~\ref{cl:opt-non}, and \ref{cl:alg-non} with Yao's principle, we get that for every (deterministic or randomized) algorithm  $\alg_\inst$, there exists an arrival order $\pi_i$ such that $$\frac{\alg_\inst(\pi_i)}{\opt_{\inst.\pi_i}(\pi_i)} \leq  \frac{1}{n} + o\left(\frac{1}{n}\right),  $$
which concludes the proof.
\end{proof}

%\section{old}
%\input{hardness-old}

\section{$k$-Uniform Matroid}
In this section we show that for the $k$-uniform feasibility constraint there is an instance in which the order-competitive ratio is $1- \frac{1}{\Theta(\sqrt{{k}})}$, which is approaching $1$ at the same rate (up to a constant) as the competitive-ratio (with respect to the prophet benchmark) for this feasibility constraint \citep{DBLP:journals/siamcomp/Alaei14,hajiaghayi2007automated}.
\begin{theorem}\label{thm:kunit}
There is a constant $c>0$ such that for every $k$, there is an instance $\inst=(\elements,\dists,\cf = \{S \subseteq \elements \mid |S| \leq k \})$ in which  
for every order-unaware algorithm $\alg_\inst$ it holds that $$\rho(\inst,\alg_\inst) \leq 1- \frac{c}{\sqrt{{k}}}.$$
\end{theorem}
\begin{proof}
Consider an instance $\inst=(\elements,\dists,\cf)$ in which $\elements = \{a_1,\ldots,a_k,b_1,\ldots,b_k,c_1\ldots,c_{2k}\}$.
The value of each element $a_i$ is deterministically $7/4$, of each element $b_i$ is deterministically $1$, and of each element $c_i$ is  either $0$ or $2$ each with probability half.

Consider the following two orders:
\begin{itemize}
    \item $\pi_1\eqdef (a_1,\ldots,a_k,b_1,\ldots,b_k,c_1\ldots,c_{2k})$
    \item $\pi_2\eqdef (a_1,\ldots,a_k,c_1\ldots,c_{2k},b_1,\ldots,b_k)$
\end{itemize} 
We first define a few notation to show an upper bound on the order-competitive ratio of this instance.
Let $X$ be the random variable of the number of non-zero values of elements $c_1,\ldots,c_{2k}$, and let $Z = \frac{k-X}{\sqrt{k/2}}$ (thus $X=k-\sqrt{\frac{k}{2}} \cdot  Z $).
We now lower bound  $\opt_{\inst,\pi_1}(\pi_1)$ and $\opt_{\inst,\pi_2}(\pi_2)$.
\begin{claim}
It holds that $$\opt_{\inst,\pi_1}(\pi_1) \geq 2k -  0.291 \sqrt{k}.$$ \label{cl:opt1}
\end{claim}
\begin{proof}
Consider an algorithm $\alg$ that selects $d \cdot \sqrt{k/2}$ elements among $\{a_1,\ldots,a_k\}$ for $d=1.152$, $0$ elements among $\{b_1,\ldots,b_k\}$, and all elements in $\{c_1,\ldots,c_{2k}\}$ with a value of $2$, as long as capacity allows.
It holds that 
\begin{eqnarray}
\opt_{\inst,\pi_1}(\pi_1) & \geq & 
\alg(\pi_1) \nonumber\\ 
& =  &\Ex\left[d \cdot \sqrt{\frac{k}{2}} \cdot \frac{7}{4} +  \min(k-d \cdot \sqrt{\frac{k}{2}},X)\cdot 2 \right]  \nonumber \\ 
& =  &\Ex\left[d \cdot \sqrt{\frac{k}{2}} \cdot \frac{7}{4} +  \min(k-d \cdot  \sqrt{\frac{k}{2}},k-\sqrt{\frac{k}{2}} \cdot  Z  )\cdot 2\right] \nonumber \\ 
& = & \Ex\left[2k-\sqrt{2k}\cdot \left(  \max(d , Z  )- d  \cdot \frac{7}{8}   \right)\right] \nonumber \\ 
& = &  2k- \Pr[Z < d] \cdot \sqrt{2k}\cdot   \frac{d}{8}  - Pr[Z\geq d ] \cdot \sqrt{2k}\cdot \Ex\left[  Z - d  \cdot \frac{7}{8}    \mid  Z \geq d \right] 
\nonumber \\ 
& \gtrsim & 2k - 0.291 \sqrt{k}, \nonumber
\end{eqnarray}
where the approximation holds since for large enough $k$, by the central limit theorem, $Z$ is approximately distributed like $N(0,1)$, and thus the result holds by the choice of  the value of $d$. 
\end{proof}

\begin{claim}
It holds that $$\opt_{\inst,\pi_2}(\pi_2) \geq 2k -  0.224 \sqrt{k}.$$ \label{cl:opt2}
\end{claim}

\begin{proof}
Consider an algorithm $\alg$ that selects $d \cdot \sqrt{k/2}$ elements among $\{a_1,\ldots,a_k\}$ for $d=0.674$, all elements in $\{c_1,\ldots,c_{2k}\}$ with a value of $2$, as long as capacity allows, and all elements among $\{b_1,\ldots,b_k\}$, as long as capacity allows.
It holds that 
\begin{eqnarray}
\opt_{\inst,\pi_2}(\pi_2) & \geq & 
\alg(\pi_2) \nonumber\\ 
& = &\Ex\left[d \cdot \sqrt{\frac{k}{2}} \cdot \frac{7}{4} +  \min(k-d \cdot \sqrt{\frac{k}{2}},X)\cdot 2 + k - d \cdot \sqrt{\frac{k}{2}} - \min(k-d \cdot \sqrt{\frac{k}{2}},X) \right]  \nonumber \\ 
& = &\Ex\left[d \cdot \sqrt{\frac{k}{2}} \cdot \frac{3}{4} +  \min(k-d \cdot  \sqrt{\frac{k}{2}},k-\sqrt{\frac{k}{2}} \cdot  Z  ) + k\right] \nonumber \\ 
& = & \Ex\left[2k-\sqrt{\frac{k}{2}}\cdot \left(  \max(d , Z  )- d  \cdot \frac{3}{4}   \right)\right] \nonumber \\ 
& = &  2k- \Pr[Z < d] \cdot \sqrt{\frac{k}{2}}\cdot   \frac{d}{4}  - Pr[Z\geq d ] \cdot \sqrt{\frac{k}{2}}\cdot \Ex\left[  Z - d  \cdot \frac{3}{4}   \mid  Z \geq d \right] 
\nonumber \\ 
& \gtrsim & 2k - 0.224 \sqrt{k}, \nonumber
\end{eqnarray}
where the approximation holds since for large enough $k$, by the central limit theorem, $Z$ is approximately distributed like $N(0,1)$, and thus the result holds by the choice of  the value of $d$. 
\end{proof}

Let $\alg_\inst$ be an arbitrary order-unaware (possibly randomized) algorithm.
Let $Y$ be the random variable that indicates the number of elements $\alg_\inst$ selects among $\{a_1,\ldots,a_k\}$ divided by $\sqrt{k/2}$.
Note that since elements $a_1,\ldots,a_k$ arrive first, $Y$ is independent on $X$.
Now for $d=0.913$ and $p=\Pr[Y > d  ]$ consider two cases: (1) $p \geq \frac{1}{2}$, and (2) $p < \frac{1}{2}$.

In case (1),  we bound the performance of $\alg_\inst$ in the case of arrival order $\pi_2$ (see Claim~\ref{cl:alg1}).
In case (2),  we bound the performance of $\alg_\inst$ in the case of arrival order $\pi_1$ (see Claim~\ref{cl:alg2}).
\begin{claim}
If $p \geq \frac{1}{2}$ then   $$\alg_{\inst} (\pi_2) \leq \frac{1}{2} \cdot \opt_{\inst,\pi_2}(\pi_2) + k -  0.115 \sqrt{k}.$$ \label{cl:alg1}
\end{claim}
\begin{proof}
It holds that
\begin{eqnarray}
\alg_{\inst} (\pi_2) & = & \Ex\left[ \alg_{\inst} (\pi_2) \mid Y > d \right] \cdot p + \Ex\left[ \alg_{\inst} (\pi_2) \mid Y \leq d  \right] \cdot \left(1 - p \right) \nonumber\\
& \leq & \frac{1}{2} \cdot \opt_{\inst,\pi_2}(\pi_2) + \frac{1}{2} \cdot \Ex\left[ \alg_{\inst} (\pi_2) \mid Y > d \right], \label{eq:alg11}
\end{eqnarray}
where the inequality is since the  algorithm conditioned on the value of $Y$, cannot obtain more than $\opt_{\inst,\pi_2}(\pi_2)$, and since $p\geq \frac{1}{2}$.
We also have that 
\begin{eqnarray}
\Ex\left[ \alg_{\inst} (\pi_2) \mid Y > d  \right] 
& \leq & k + \Ex\left[ \sqrt{\frac{k}{2}} \cdot Y \cdot \frac{3}{4} +   \min(k-\sqrt{\frac{k}{2}} \cdot Y   ,X)\mid Y > d  \right] \nonumber\\
%& = & k+ \Ex\left[ Y \cdot \frac{3}{4} +  \min(k-Y,X) \mid Y > d \cdot  \sqrt{\frac{k}{2}} \right] \nonumber\\
%& = & k+\Ex\left[ Y \cdot \frac{3}{4} +  \min(k-Y,k-\sqrt{\frac{k}{2}} \cdot  Z) \mid Y > d \cdot  \sqrt{\frac{k}{2}} \right] \nonumber\\
& = &  2k-\sqrt{\frac{k}{2}}\cdot\Ex\left[ \max( Y , Z  )- Y  \cdot \frac{3}{4}    \mid Y > d  \right] \nonumber \\ 
& \leq  &  2k-\sqrt{\frac{k}{2}}\cdot\Ex\left[ \max( d , Z  )- d  \cdot \frac{3}{4}     \right] \nonumber \\ 
%& = &  \Ex\left[2k-\sqrt{\frac{k}{2}}\cdot \left(  \max(d , Z  )- d  \cdot \frac{3}{4}   \right)\right] \nonumber \\ 
& = & 2k - \Pr[Z < d] \cdot \sqrt{\frac{k}{2}}\cdot   \frac{d}{4}  -   Pr[Z\geq d ] \cdot \sqrt{\frac{k}{2}}\cdot \Ex\left[  Z - d  \cdot \frac{3}{4}   \mid  Z \geq d \right] 
\nonumber \\ 
& \lesssim &  2k - 0.231 \sqrt{k},  \label{eq:alg12}
\end{eqnarray}
where the first inequality is since  the   value obtained by the algorithm can be bounded in the following way: first the algorithm receives $1$ for each selected box, it then receives  an additional term of $\frac{3}{4}$ for each selected box in $\{a_1,\ldots,a_k\}$, and an additional term of $1$ for each selected box in $\{c_1,\ldots,c_{2k}\}$ with a value of $2$.
The first equality holds by rearranging and replacing $X$ by $k-\sqrt{\frac{k}{2}} \cdot  Z $. The second inequality holds since the function $f(x)= \Ex\left[ \max( x , Z  )- x  \cdot \frac{3}{4} \right] $ is an increasing function in $x$ for $x>d$. The last inequality holds for large enough $k$  by the central limit theorem.
Combining Equations~\eqref{eq:alg11} and \eqref{eq:alg12} concludes the proof.
\end{proof}

\begin{claim}
If $p < \frac{1}{2}$ then   $$\alg_{\inst} (\pi_1) \leq \frac{1}{2} \cdot \opt_{\inst,\pi_1}(\pi_1) + k -  0.150 \sqrt{k}.$$ \label{cl:alg2}
\end{claim}
\begin{proof}
It holds that
\begin{eqnarray}
\alg_{\inst} (\pi_1) & = & \Ex\left[ \alg_{\inst} (\pi_1) \mid Y > d \right] \cdot p + \Ex\left[ \alg_{\inst} (\pi_1) \mid Y \leq d  \right] \cdot \left(1 - p \right) \nonumber\\
& \leq & \frac{1}{2} \cdot \opt_{\inst,\pi_1}(\pi_1) + \frac{1}{2} \cdot \Ex\left[ \alg_{\inst} (\pi_1) \mid Y \leq d \right], \label{eq:alg21}
\end{eqnarray}
where the inequality is since the  algorithm conditioned on the value of $Y$, cannot obtain more than $\opt_{\inst,\pi_1}(\pi_1)$, and since $p < \frac{1}{2}$.
We are now going to bound $\Ex\left[ \alg_{\inst} (\pi_1) \mid Y \leq d \right]$. To do so, we observe that the optimal online algorithm  that already selected $\sqrt{\frac{k}{2}} \cdot Y$ elements among $\{a_1,\ldots,a_k\}$  and knows that the arrival order is $\pi_1$, is a deterministic algorithm. Moreover, the optimal algorithm never selects elements among $\{b_1,\ldots,b_k\}$. This is since selecting such element  increases the algorithm's value  by $1$ when $Z > Y $, but decreases the algorithm's value by $1$, when $Z \leq Y$. It follows then by the fact that the probability that $Z > Y$ for every non-negative $Y$ is at most $\frac{1}{2}$.
%since it will receive a value of $1$ if  and lose a value of $2$. 
Therefore, % we can bound from above  $\alg_\inst(\pi_1)$ as follows: 
\begin{eqnarray}
\Ex\left[ \alg_{\inst} (\pi_1) \mid Y \leq d  \right] 
& \leq & \Ex\left[ \sqrt{\frac{k}{2}} \cdot Y \cdot \frac{7}{4} +  \min(k-\sqrt{\frac{k}{2}} \cdot Y   ,X) \cdot 2 \mid Y \leq d  \right] \nonumber\\
& = &  2k-\sqrt{2k}\cdot\Ex\left[ \max( Y, Z  )- Y  \cdot \frac{7}{8}    \mid Y \leq d  \right] \nonumber \\ 
& \leq  &  2k-\sqrt{2k}\cdot\Ex\left[ \max( d  , Z  )- d  \cdot \frac{7}{8}     \right] \nonumber \\ 
& = & 2k - \Pr[Z < d] \cdot \sqrt{2k}\cdot   \frac{d}{8}  -  Pr[Z\geq d ] \cdot \sqrt{2k}\cdot \Ex\left[  Z - d  \cdot \frac{7}{8}    \mid  Z \geq d \right] 
\nonumber \\ 
& \lesssim &  2k - 0.301 \sqrt{k},  \label{eq:alg22}
\end{eqnarray}
%where the first inequality is since  the   value obtained by the algorithm can be bounded in the following way: the algorithm receives $\frac{7}{4}$ for each selected box in $\{a_1,\ldots,a_k\}$, and $2$ for each selected box in $\{b_1,\ldots,b_{2k}\}$ with a value of $2$.
The first equality holds by rearranging and replacing $X$ by $k-\sqrt{\frac{k}{2}} \cdot  Z $. 
The second inequality holds 
%since if $w_Y >0$, and $Y<d$, then if we decrease $w_Y$ and increase $Y$ by the same term, then the expression is increased, and for the case that $w_Y=0$ it holds 
since the function $f(x)= \Ex\left[ \max( x , Z  )- x  \cdot \frac{7}{8} \right] $ is a decreasing function in $x$ for $x \leq d$. The last inequality holds for large enough $k$  by the central limit theorem. 
Combining Equations~\eqref{eq:alg21} and \eqref{eq:alg22} concludes the proof.
\end{proof}
The proof then follows by considering the two mentioned cases:
If $p \geq  \frac{1}{2}$ then when considering $\pi_2$, we get that
\begin{eqnarray}
\frac{\alg_{\inst} (\pi_2)}{\opt_{\inst,\pi_2}(\pi_2)} 
& \leq & \frac{\frac{1}{2} \cdot \opt_{\inst,\pi_2}(\pi_2) + k -  0.115 \sqrt{k}}{\opt_{\inst,\pi_2}(\pi_2)} \nonumber \\
& = & \frac{1}{2} + \frac{k -  0.115 \sqrt{k}}{\opt_{\inst,\pi_2}(\pi_2)} \nonumber \\
& \leq & \frac{1}{2} + \frac{k -  0.115 \sqrt{k}}{2k -  0.224 \sqrt{k}} \leq 1- \frac{0.001}{\sqrt{k}}, \nonumber
\end{eqnarray}
where the first inequality is by Claim~\ref{cl:alg1}, and the second inequality is by Claim~\ref{cl:opt2}.

If $p< \frac{1}{2}$ then when considering $\pi_1$, we get that
\begin{eqnarray}
\frac{\alg_{\inst} (\pi_1)}{\opt_{\inst,\pi_1}(\pi_1)} 
& \leq & \frac{\frac{1}{2} \cdot \opt_{\inst,\pi_1}(\pi_1) + k -  0.150 \sqrt{k}}{\opt_{\inst,\pi_1}(\pi_1)} \nonumber \\
& = & \frac{1}{2} + \frac{k -  0.150 \sqrt{k}}{\opt_{\inst,\pi_1}(\pi_1)} \nonumber \\
& \leq & \frac{1}{2} + \frac{k -  0.150 \sqrt{k}}{2k -  0.291 \sqrt{k}} \leq 1- \frac{0.002}{\sqrt{k}}, \nonumber
\end{eqnarray}
where the first inequality is by Claim~\ref{cl:alg2}, and the second inequality is by Claim~\ref{cl:opt1}.
\end{proof}

\section{Open Problems}
Our goal in this paper was to ask whether, with respect to the new benchmark of the order-competitive ratio, it is possible to achieve better asymptotic results than with respect to the traditional competitive-ratio.

One natural Open question is whether in settings where the best competitive-ratio is half, it is possible to achieve a better than half order-competitive ratio. \ificalp{Ezra et al.}\citet{ezra2023next} showed that this is possible for single-choice prophet inequality, but for many other feasibility constraints (e.g., matching, matroids, knapsack, etc.), this is still an open question.
Another open question is what is the best order-competitive ratio or competitive-ratio for the family downward-closed feasibility constraints, and whether they are the same. The best known lower bound on the competitive-ratio (and also the order-competitive ratio) is $O\left(\frac{1}{log^{2}(n)}\right)$ by \ificalp{Rubinstein}\citet{DBLP:conf/stoc/Rubinstein16}.

%\section{Non-Downward Closed Feasibility Constraints iid}
%\input{iid}

\section*{Acknowledgments}
This project is supported by the ERC Advanced Grant 788893 AMDROMA, EC H2020RIA project “SoBigData++” (871042), PNRR MUR project PE0000013-FAIR”, PNRR MUR project  IR0000013-SoBigData.it.

\bibliographystyle{abbrvnat}
\bibliography{bib}

\end{document}